\documentclass[12pt]{article}
\usepackage{authblk}
\usepackage[UKenglish]{babel}
\usepackage[UKenglish]{isodate}
\usepackage[utf8]{inputenc}
\usepackage{enumitem}
\usepackage{amsmath,amsthm,amsfonts,amssymb}
\usepackage{tikz}
\usepackage{tikz-cd}
\usepackage{stmaryrd}  
\usepackage{a4wide}
\usetikzlibrary{positioning,decorations.pathmorphing,shapes,calc}
\usepackage{hyperref}
\usepackage{natbib}
\usepackage{mathrsfs} 
\usepackage{proof}

\providecommand{\noopsort}[1]{}

\newcommand*\cocolon{%
        \nobreak
        \mskip6mu plus1mu
        \mathpunct{}%
        \nonscript
        \mkern-\thinmuskip
        {:}%
        \mskip2mu
        \relax
}

\makeatletter
\def\thanks#1{\protected@xdef\@thanks{\@thanks
        \protect\footnotetext{#1}}}
\makeatother

\theoremstyle{plain}
\newtheorem{theorem}{Theorem}[section]
\newtheorem{proposition}[theorem]{Proposition}

\newtheorem{corollary}[theorem]{Corollary}

\newtheorem*{claim*}{Claim}

\theoremstyle{definition}
\newtheorem{definition}[theorem]{Definition}

\newcommand{\mono}{\hookrightarrow} 
\newcommand{\epi}{\twoheadrightarrow} 

\newcommand{\id}{\mathrm{id}}

\newcommand{\Sg}{\ensuremath{\sigma}}
\newcommand{\FO}{\ensuremath{\mathrm{FO}}}
\newcommand{\FOA}{\ensuremath{\mathrm{FO}_\omega}}

\newcommand{\Mod}{\ensuremath{\mathrm{Mod}}}
\newcommand{\ModA}{\ensuremath{\mathrm{Mod}_\omega}}

\newcommand{\TypA}{\ensuremath{\mathrm{Typ}_\omega}}
\newcommand{\Fin}{\ensuremath{\mathrm{Fin}}}

\newcommand{\F}{\ensuremath{\mathcal{F}}}
\newcommand{\M}{\ensuremath{\mathcal{M}}}
\newcommand\MA{\ensuremath{\mathbb{M}}}
\newcommand{\V}{\ensuremath{\mathcal{V}}}
\renewcommand{\P}{\ensuremath{\mathcal{P}}}
\newcommand{\G}{\ensuremath{\mathbf{\Gamma}}}
\newcommand\indi{\ensuremath{\delta}} 

\newcommand{\two}{\ensuremath{\mathbf{2}}}

\newcommand{\w}{\widehat}

\newcommand{\N}{\mathbb{N}}

\newcommand\SP[1]{\ensuremath{\left<#1\right>}}
\newcommand\ARG{\mkern1.5mu\text{-}\mkern1.5mu}
\newcommand\PrG[1]{\mathbf{p}_{\geq #1}\,} 
\newcommand\PrL[1]{\mathbf{p}_{< #1}\,}
\newcommand{\mm}{^{-}} 
\newcommand{\cc}{^\circ} 

\renewcommand{\S}{\ensuremath{\mathbf{S}}}
\newcommand{\Set}{\ensuremath{\mathbf{Set}}}
\newcommand{\BStone}{\ensuremath{\mathbf{BStone}}}
\newcommand{\BA}{\ensuremath{\mathbf{BA}}}
\newcommand{\op}{\ensuremath{\mathrm{op}}}
\newcommand{\Con}{\ensuremath{\mathbf{Con}}}
\newcommand{\Posf}{\ensuremath{\mathbf{Pos}_f}}

\newcommand\NOdM{Ne\v{s}et\v{r}il and Ossona de Mendez}

\renewcommand{\phi}{\varphi}
\renewcommand{\epsilon}{\varepsilon}
\newcommand{\sem}[1]{\ensuremath{\llbracket #1 \rrbracket}}

\newcommand\inv{^{-1}}

\newcommand\qtq[1]{{\quad\text{#1}\quad}}
\newcommand\ete[1]{{\enspace\text{#1}\enspace}}
\newcommand\ee[1]{{\enspace#1\enspace}}

\newcommand\p[1]{\ensuremath{\mathcal #1}}

\title{A Cook's tour of duality in logic: from quantifiers, through Vietoris, to measures\thanks{This project has been funded by the European Research Council (ERC) under the European Union's Horizon 2020 research and innovation program (grant agreement No.670624). 
Tom\'a\v s Jakl has received partial support from the EPSRC grant EP/T007257/1.
Luca Reggio has received funding from the European Union's Horizon 2020 research and innovation programme under the Marie Sk{\l}odowska-Curie grant agreement No.837724.}}
\author[1]{Mai Gehrke}
\author[1,2]{Tom\'a\v s Jakl}
\author[3]{Luca Reggio}
\affil[1]{CNRS and Universit{\'e} C{\^o}te d'Azur, Nice, France}
\affil[2]{Department of Computer Science and Technology, University of Cambridge, UK}
\affil[3]{Department of Computer Science, University of Oxford, UK}
\date{}

\begin{document}

\maketitle

\vspace{-2em}
\abstract{We identify and highlight certain landmark results in Samson Abramsky's work which we believe are fundamental to current developments and future trends. 
In particular, we focus on the use of 
\begin{itemize}
\item topological duality methods to solve problems in logic and computer science;
\item category theory and, more particularly, free (and co-free) constructions;
\item these tools to unify the `power' and `structure' strands in computer science.
\end{itemize}
}

\section{Algebras from logic}
\label{s:algebras-from-logic}

Boole wanted to view propositional logic as arithmetic. This idea, of seeing logic as a kind of algebra, reached a broader and more foundational level with the work of Tarski and the Polish school of algebraic logicians. The basic concept is embodied in what is now known as the \emph{Lindenbaum-Tarski algebra} of a logic. In the classical cases, this algebra is obtained by quotienting the set of all formulas \p F by logical equivalence, that is,
\[ \p L = \p F/_{\approx} \qtq{where} \phi \approx \psi \ete{if, and only if,} \phi \ete{and} \psi \ete{are logically equivalent.}\]
When the equivalence relation $\approx$ is a congruence for the connectives of the logic, \p L may be seen as an algebra in the signature given by the connectives. This is the case for many propositional logics as well as for first-order logic. There is, however, a fundamental difference in how well this works at these two levels of logic. 

For example, for Classical Propositional Logic (CPL), Intuitionistic Propositional Calculus (IPC) and modal logics, the Lindenbaum-Tarski algebra is the \emph{free algebra} over the set of primitive propositions of the appropriate variety. In the above mentioned cases, these are Boolean algebras, Heyting algebras, and modal algebras of the appropriate signature, respectively.  
Further, for algebras in these varieties, congruences are given by the equivalence classes of the top elements which, logically speaking, are the theories of the corresponding logics. Consequently, we have that the Lindenbaum-Tarski algebras of theories, in which one quotients out by logical equivalence modulo the theory, account for the full varieties of Boolean algebras, Heyting algebras and modal algebras.

The picture is not always quite this simple, even at the propositional level. E.g.\ the Lindenbaum-Tarski algebra of positive propositional logic (i.e.\ the fragment of CPL without negation, which we will denote PPL) is indeed the free bounded distributive lattice over the set of primitive propositions. However, since there are lattices with multiple congruences giving the same filter, we do not have the same natural correspondence between the full variety of distributive lattices and the theories of PPL.  This sort of problem can be dealt with and this is the subject of the far-reaching theory of Abstract Algebraic Logic, see \citep{Font91} for the example of PPL.

\vspace{1em}
Let us now consider (classical) first-order logic. Here also, logical equivalence is a congruence for the logical connectives. We have the Boolean connectives, and unary connectives $\exists x$ and $\forall x$, a pair for each individual variable $x$ of the logical language.\footnote{Typically one also considers some named constants, which we are not mentioning here.} The latter give rise to pairs of unary operations that are inter-definable by conjugation with negation. Thus, in the Boolean setting, it is enough to consider the  $\exists x$ operations. These are (unary) \emph{modal operators}.
 
In its most basic form, modal propositional logic corresponds to the variety of modal algebras (MAs), which are Boolean algebras augmented by a unary operation that preserves finite joins. The algebraic approach is a powerful tool in the study of modal logics, see e.g.\ \citep{RWZ06} for a survey. In particular, the Lindenbaum-Tarski algebra for this logic is the free modal algebra over the propositional variables, the normal modal logic extensions correspond to the subvarieties of the variety of MAs, and theories in these logics correspond to the individual algebras in the corresponding varieties.

The Lindenbaum-Tarski algebra of first-order formulas modulo logical equivalence is a multimodal algebra, with modalities $\Diamond_x$, one for each variable $x$ in the first-order language.  These modalities satisfy some equational properties such as\footnote{Throughout, if no confusion arises, we write $\phi$ for the corresponding element of the Lindenbaum-Tarski algebra, i.e.\ the logical equivalence class $[\phi]_{\approx}$ of the formula $\phi$.} 
\[
\varphi\leq\Diamond_x\varphi\quad\quad \Diamond_x (\varphi \wedge \Diamond_x \psi) = \Diamond_x \varphi\wedge\Diamond_x\psi \quad\quad \Diamond_x \Diamond_y \varphi = \Diamond_y \Diamond_x\varphi.
\]
A fundamental problem, as compared with the propositional examples given above, is that these Lindenbaum-Tarski algebras \emph{are not free} in any reasonable setting. Tarski and his students introduced the variety of cylindric algebras of which these are examples, see \citep{Mo86} for an overview. However, not all cylindric algebras occur as Lindenbaum-Tarski algebras for first-order theories.  For one, when we have an infinite set of variables, and thus of modalities, for every element $\varphi$ in the algebra there is a finite set $V_\varphi$ of variables such that $\Diamond_x\varphi = \varphi$ for all $x\not\in V_\varphi$.

Even though cylindric algebras have been extensively studied, little is known specifically about the ones arising as Lindenbaum-Tarski algebras of first-order theories. A notable exception is the paper \citep{Myers76} characterising the algebras for first-order logic over empty theories. Another important insight, due to Rasiowa and Sikorski, is the fact that the completeness theorem for first-order logic may be obtained using the Lindenbaum-Tarski construction \citep{RS50}. Their proof uses the famous Rasiowa-Sikorski Lemma. This lemma, which may be seen as a consequence of the Baire Category Theorem in topology, states that, given a specified countable collection of subsets with suprema in a Boolean algebra, one can separate the elements of the Boolean algebra with ultrafilters that are inaccessible by these suprema.
 
\vspace{1em}
The lack of freeness of the Lindenbaum-Tarski algebras of first-order logic is overcome by moving from lattices with operators to categories and categorical logic. 
In the equational setting, algebraic theories can equivalently be described as Lawvere theories, i.e.\ categories with finite products and a distinguished object $X$ such that every object is a finite power of $X$.\footnote{For a variety of algebras $\mathscr{V}$, the associated Lawvere theory is the \emph{dual} of the category of finitely generated free $\mathscr{V}$-algebras with homomorphisms; the distinguished object is the free algebra on one generator.} Similarly, theories in a given fragment of first-order logic correspond to a certain class of categories. 

For instance, theories in the positive existential fragment of first-order logic, also called coherent theories, correspond to coherent categories.
Every coherent theory $T$ yields a coherent category, the \emph{syntactic category} of $T$, which may be seen as a generalisation of the Lindenbaum-Tarski construction, and which is free in an appropriate sense. Central to this construction is the fundamental insight, of Lawvere, that quantifiers are adjoints to substitution maps. Thus, existential quantifiers are encoded in coherent categories as lower adjoints to certain homomorphisms between lattices of subobjects. Further, there is some sense in which the correspondence between theories and quotients is regained (at the level of so-called classifying toposes of the theories). See \citep{MR1977}. 
Other fragments of first-order logic can be dealt with in a similar fashion, e.g.\ intuitionistic first-order theories correspond to Heyting categories, and classical first-order theories to Boolean coherent categories. See \citep{ElephantV2} for a thorough exposition.

To make the relation between syntactic categories and Lindenbaum-Tarski algebras more explicit, we recall the notion of Boolean hyperdoctrines, tightly related to Boolean coherent categories.
Consider the category $\Con$ of contexts and substitutions. A context is a finite list of variables $\overline{x}$, and a substitution from $\overline{x}$ to a context $\overline{y}=y_1,\ldots,y_n$ is a tuple $\langle t_1,\ldots, t_n\rangle$ of terms with free variables in $\overline{x}$. Given a first-order theory $T$, let $P(\overline{x})$ be the Lindenbaum-Tarski algebra of first-order formulas with free variables in $\overline{x}$, up to logical equivalence modulo $T$. A substitution $\langle t_1,\ldots, t_n\rangle\colon \overline{x}\to\overline{y}$ induces a Boolean algebra homomorphism $P(\overline{y})\to P(\overline{x})$ sending a formula $\phi(\overline{y})$ to $\phi(\langle t_1,\ldots, t_n\rangle/\overline{y})$.\footnote{More precisely, the morphisms in $\Con$ are defined as equivalence classes of substitutions, by identifying two tuples $\langle s_1,\ldots, s_n\rangle$ and $\langle t_1,\ldots, t_n\rangle$ if they give rise to the same homomorphism.} This yields a functor 
\begin{equation*}
P\colon \Con^\op\to \BA.
\end{equation*}
The product projection $\pi_y\colon\overline{x},y\to \overline{x}$ in $\Con$ induces the Boolean algebra embedding $P(\pi_y)\colon P(\overline{x})\hookrightarrow P(\overline{x},y)$, which admits both lower and upper adjoints:
\begin{gather*}
\exists_{y}\dashv P(\pi_y), \ \ \exists_{y}(\phi(\overline{x},y))=\exists y.\phi(\overline{x},y), \\
P(\pi_y)\dashv \forall_{y}, \ \ \forall_{y}(\phi(\overline{x},y))=\forall y.\phi(\overline{x},y).
\end{gather*}
This accounts for the \emph{Boolean hyperdoctrine} structure of $P$. The syntactic category of the theory $T$ can be obtained from $P$ by means of a 2-adjunction between Boolean hyperdoctrines and Boolean categories, cf.\ \citep{Pitts1983} or \citep[Chapter~5]{Coumans2012}.

While the categorical perspective solves a number of problems, it is not easily amenable to the inductive point of view that we want to highlight here. We will get back to this in Section~\ref{s:three-ex-spaces}.

\section{Topological methods in logic}
\label{s:top-methods-in-logic}

Topological methods in logic have their origin in the work of M.\ H.\ Stone. The paper \citep{Stone1936} established what is nowadays presented as a dual equivalence between the category $\BA$ of Boolean algebras with homomorphisms and a full subcategory $\BStone$ of the category of topological spaces with continuous maps. The objects of $\BStone$ are the so-called \emph{Boolean (Stone) spaces}, i.e.\ compact Hausdorff spaces whose collection of \emph{clopen} (simultaneously closed and open) subsets forms a basis for the topology. Usually referred to as \emph{Stone duality for Boolean algebras}, this is the prototypical example of a dual equivalence induced by a dualizing object, i.e.\ an object sitting at the same time in two categories. In fact, the quasi-inverse functors providing the equivalence between $\BA^\op$ and $\BStone$ are given by enriching the set of homomorphisms into the appropriate structure on the two-element set $\two=\{0,1\}$, which can be seen either as the two-element Boolean algebra or as the two-element Boolean space when equipped with the discrete topology. 

Given a Boolean algebra $B$, the space $X_B$ obtained by equipping the set of homomorphisms
\[
\hom_{\BA}(B,\two)
\]
with the subspace topology induced by the product topology on $\two^B$ is a Boolean space, the \emph{(Stone) dual space} of $B$. Under the correspondence sending a Boolean algebra homomorphism $h\colon B\to\two$ to the subset $h^{-1}(1)\subseteq B$, the points of $X_B$ can be identified with the \emph{ultrafilters} on $B$. In logical terms, these  are the complete consistent theories over $B$. 
Conversely, given a Boolean space $X$, the set of continuous~maps
\[
\hom_{\BStone}(X,\two)
\]
forms a Boolean subalgebra $B_X$ of the product algebra $\two^X$, where $\two$ is now viewed as a Boolean algebra. When equipped with the induced Boolean operations, $B_X$ is called the \emph{dual algebra} of $X$. Upon identifying a continuous function $f\colon X\to\two$ with the clopen subset $f^{-1}(1)\subseteq X$, the Boolean algebra $B_X$ can be described as the field of clopen subsets of $X$ with the set-theoretic Boolean operations. Stone duality states that these object assignments extend to functors, and there are isomorphisms $B\cong B_{X_{B}}$ and $X\cong X_{B_{X}}$ (natural in $B$ and $X$, respectively). Throughout, the element of $B_{X_{B}}$ corresponding to $a\in B$ will be denoted by $\w{a}$.

Shortly after his seminal work in 1936, Stone generalised the duality to bounded distributive lattices \citep{Stone1938}; there, the relevant category of spaces consists of spectral spaces with perfect maps. A different formulation of the duality for distributive lattices, induced by the dualizing object $\two$ regarded either as a lattice or as a discrete \emph{ordered} space where $0<1$, was later introduced in \citep{Priestley1970}.

When combined with the algebraic semantics, as outlined in the previous section, Stone duality yields a powerful framework for developing and applying topological methods in logic.
The potential advantages of applying duality are of two types. For one, duality theory often connects syntax and semantics. To wit, in the case of CPL, the Lindenbaum-Tarski algebra is the free Boolean algebra on the set $V$ of propositional variables, and its dual space is the Cantor space $\two^V$ of all valuations over $V$. The second type of advantage is that it \emph{often is easier}, technically, to solve a problem on the dual side.

The use of duality is not restricted to the Boolean setting. Indeed, generalisations and extensions of Stone duality have been exploited to study fragments and extensions of CPL. Many other special cases have since been developed based on Stone's and Priestley's dualities for bounded distributive lattices (corresponding to PPL). Here we just mention the duality for Heyting algebras, the algebraic semantics of IPC, mainly developed by Leo Esakia \citep{Esakia1974,Esakia2019}.
Stone duality was also extended by J\'onsson and Tarski to Boolean algebras with operators by introducing the powerful framework of canonical extensions \citep{JT1,JT2}. This was a crucial step for many applications, e.g.\ in modal logic.

\vspace{1em}
In theoretical computer science, the link between syntax and semantics provided by Stone-type dualities is particularly central as the two sides correspond to specification languages and to spaces of computational states, respectively. The ability to translate faithfully between these two worlds has often proved itself to be a powerful theoretical tool as well as a handle for solving problems. A prime example is Abramsky's seminal work \citep{Abramsky87,Abramsky91} linking program logic and domain theory via Stone duality for bounded distributive lattices,  which was awarded the IEEE LICS ``Test of Time'' Award in 2007. Other examples include large parts of modal and intuitionistic logics, where J\'onsson-Tarski duality yields Kripke semantics \citep{BlackburnDeRijkeVenema2001}. For a particular example, see  Ghilardi's work in modal and intuitionistic logic on unification \citep{Ghi2004} and normal forms \citep{Ghi1995}. 

By contrast, Stone duality has not played a significant role, at least overtly, in more algorithmic areas of theoretical computer science until recently. In the theory of regular languages,  finite and profinite monoids are an important tool, in particular for proving decidability, ever since their introduction in the 1960s and 1980s, respectively, see \citep{Pin09} for a survey.  While it was observed as early as 1937 by Birkhoff that profinite topological algebras are based on Boolean spaces \citep{Birkhoff1937}, the connection with Stone duality was not used in automata theory until much more recently. It was exploited first in an isolated case by \citep{Pippenger97}, and then more structurally by \citep{GGP2008}. Further, realising that these methods are instances of Stone duality provides an opportunity to generalise them to the setting of computational complexity and the search for lower bounds \citep{GK2017}. This line of work connects tools from semantics, such as Stone duality, with problems and methods on the algorithmic side of computer science, such as decidability and Eilenberg-Reiterman theory. Similarly, recent work of Samson Abramsky and co-workers connects categorical tools from semantics, such as comonads, with concepts from finite model theory, such as tree-width and tree-depth \citep{Abramsky2017b, AbramskyShah2018}.

Finite model theory, computational complexity theory and the theory of regular languages all belong to the branch of computer science where the use of resources in computing is the main focus, whereas category theory and Stone duality have long been central tools in semantics of programming languages. While the trend of making connections and seeking unifying results that bridge the gap between semantics and algorithmic issues has long been on the way (e.g.\ in the form of semantic work on resource sensitive logics), making this overt and placing it front and center stage is a recent phenomenon in which Samson Abramsky has played a central role. In particular, one may mention the 2017 semester-long program at The Simons Institute for the Theory of Computing on Logical Structures in Computation of which he was a co-organiser, and the ensuing work and ongoing project with Anuj Dawar focussing on bridging what they aptly call the \emph{Structure versus Power} gap in theoretical computer science. The 2014 ERC project Duality in Formal Languages and Logic -- a unifying approach to complexity and semantics (DuaLL), in which our recent work has taken place, shares these goals. 

In Section~\ref{s:modal-Vietoris}, we highlight some of the ideas and concepts from Samson Abramsky's work in semantics that are playing an important role in our recent work on the DuaLL project, which we will describe in Section~\ref{s:Quant}. In Section~\ref{s:three-ex-spaces}, we briefly review two settings from logic pertinent to our work, and give a duality-centric description of the treatment of the function space construction in Abramsky's Domain Theory in Logical Form. This allows us to make a connection to the profinite methods in automata theory.

\subsection{Modal logic and the Vietoris functor}\label{s:modal-Vietoris}
An important contribution of Samson Abramsky's is to use the duality between syntax and semantics, \emph{combined with a step-wise description of connectives} in logic applications. This phenomenon is the driving force behind his sweeping and elegant general solution to domain equations in the paper Domain Theory in Logical Form (DTLF), \citep{Abramsky91}. We will get back to this with a few more details in Section~\ref{s:three-ex-spaces}. In  \citep{Abramsky1988}, which is the published version of various talks given during the genesis of DTLF, Abramsky gives a simpler example of this general idea. The setting is non-well-founded sets, and the object he considers is the free modal algebra (over the empty set). Other early uses of similar methods are due to Ghilardi \citep{Ghilardi1992,Ghi1995}. Subsequently, the treatment of the free modal algebra given in Abramsky's talks, in particular his talk at the 1988 British Colloquium on Theoretical Computer Science in Edinburgh, has been identified as an important contribution to modal logic in its own right, see e.g.\ \citep{Ruttenetal93,KuKuVe2004,VV2014}, and it is also very pertinent to the duality theoretic treatment of quantifiers which we will discuss in Section~\ref{s:Quant}.

The step-wise description of an algebra from a set of generators is what is often called \emph{Noetherian induction} in algebra and \emph{induction on the complexity of a formula} in logic: The algebra is generated layer by layer, starting with the generators --- which are said to be of rank $0$ --- by adding consecutive layers of the operations to obtain higher rank elements. Also, instead of doing this with all the operations, we may do it relative to a fragment. In the case of modal algebras, for example, we may consider as rank $0$ all Boolean combinations of generators, rank less than or equal to $1$ any element which may be expressed as a Boolean combination of rank $0$ and diamonds of rank $0$ elements, and so on. This is a fine tool for the purpose of induction, but it is not a good tool for constructing algebras in general. However, if the operation is freely added modulo some equations which are of pure rank $1$, then it is in fact a powerful method of \emph{construction}. This is exactly the situation for free modal algebras, which are Boolean algebras with an additional operation satisfying the equations
\[   \Diamond 0 \approx 0 \qtq{and} \Diamond(x\vee y) \approx \Diamond x \vee \Diamond y. \]
These equations are both of pure rank $1$. That is, in each equation, all occurrences of each variable are in the scope of exactly \emph{one} layer of modal operators. 

From a categorical point of view, one may see algebras in a variety as Eilenberg-Moore algebras for a finitary monad, but having a pure rank $1$ axiomatisation means that these are also presentable as the \emph{algebras for an endofunctor}, see \citep{KurzRosicky12} where this is studied in greater generality. In the case of MAs, define the endofunctor $\MA$ on Boolean algebras which takes a Boolean algebra $B$ to the Boolean algebra freely generated by elements $\Diamond a$, for every $a\in B$, subject to the equations for modal algebras viewed as relations on these generators:
\[   \Diamond 0 \approx 0 \qtq{and} \Diamond(a\vee b) \approx \Diamond a \vee \Diamond b \quad(\forall a,b\in B).\]
Then $B$, equipped with a unary operation $f\colon B\to B$, is a modal algebra if and only if the map $\Diamond a\mapsto f(a)$ extends to a Boolean algebra homomorphism $h\colon \MA(B)\to B$. It also follows that the free modal algebra over a Boolean algebra $B$ may be \emph{constructed inductively}, as the colimit of the  sequence
\[\begin{tikzcd}
B_0 \arrow[hookrightarrow]{r}{i_0} & B_1 \arrow[hookrightarrow]{r}{i_1} & B_2\arrow[hookrightarrow]{r}{i_2} & \dots 
\end{tikzcd}\]
where $B_0=B$, $B_{n+1}$ is the coproduct $B\oplus\MA(B_n)$, the map $i_0$ is the embedding of $B$ in the coproduct, and $i_{n+1}=\id_{B} \oplus \MA(i_n)$. Note that, if $B$ is finite, then so are all the algebras in the sequence. Moreover, if we start with the free Boolean algebra on a set $V$, then the colimit of the sequence is the free modal algebra over $V$, and $B_n$ is the Boolean subalgebra consisting of all formulas of rank at most $n$. 
 
Further, we may of course dualize $\MA$ to get a functor on $\BStone$ and a co-inductive description of the dual of free modal algebras. This dual endofunctor is the Vietoris functor. Recall that, given a Boolean space $X$, the \emph{Vietoris hyperspace of $X$} is the collection $\V(X)$ of closed subsets of $X$ equipped with the topology generated by the sets of the form
\[ 
\Diamond U = \{ C \in \V(X) \mid C \cap U \neq \emptyset \} \qtq{and} (\Diamond U)^c 
\]
for $U$ a clopen subset of $X$. With respect to this topology, $\V(X)$ is again a Boolean space. See \citep{Vietoris1923,Michael1951}. Furthermore, for every continuous map $f\colon X\to Y$, the forward-image map $f(\ARG)\colon \V(X)\to\V(Y)$ is continuous. Hence, we obtain a functor
\[
\V\colon \BStone\to\BStone.
\]
Abramsky showed that the dual Stone space of the free modal algebra on no generators coincides with the final coalgebra for the functor $\V$. In general, the dual of the sequence of embeddings given above is
\[\begin{tikzcd}[column sep=4.6em]
X \arrow[twoheadleftarrow]{r}{\pi_X} &  X\times\V(X)=X_1 \arrow[twoheadleftarrow]{r}{\ \id_X\times\V(\pi_X)} & X\times\V(X_1)=X_2  \arrow[twoheadleftarrow]{r} & \dots
\end{tikzcd}\]
This result provides also a coalgebraic perspective on the duality between modal algebras and descriptive general Kripke frames. As such, it has had a strong influence on the very active coalgebraic approach to modal logic. The Vietoris hyperspace construction also appeared earlier in modal logic in the work (published in Russian) of Leo Esakia, cf.\ \citep{Esakia1974}. See also \citep{Esakia2019} for the recent English translation of Esakia's 1985 book.

\subsection{Three examples of dual spaces in logic}\label{s:three-ex-spaces}
In this section we discuss duality methods in logic in three settings: classical first-order logic, B{\"u}chi's logic on words, and Domain Theory in Logical Form. 

\paragraph{First-order logic and spaces of types.}
For classical first-order logic, the dual space of the Lindenbaum-Tarski algebra of formulas is fairly easy to describe. Fix a countably infinite set of first-order variables $v_1,v_2,\ldots$ and a first-order signature $\Sg$, i.e.\ $\Sg$ may contain relation symbols as well as function symbols and constants. Denote by $\FOA$ the set of all first-order formulas in the signature $\Sg$ over the set of variables. Given a theory $T$, that is, any set of first-order sentences in the signature $\Sg$, consider the collection
\[
\ModA(T)=\{(A,\alpha\colon \omega\to A)\mid A \ \text{is a $\Sg$-structure and} \ A\models T\}
\]
of models of $T$ equipped with an assignment of the variables. 
The satisfaction relation ${\models}\subseteq \ModA\times\FOA$ induces the equivalence relations of elementary equivalence and logical equivalence on these sets, respectively:
\[
(A,\alpha)\equiv (A',\alpha') \ \ \text{ iff } \ \ \forall \phi\in\FOA \ \ A,\alpha\models \phi \ \Longleftrightarrow \ A',\alpha'\models \phi
\]
and
\[
\phi\approx \psi \ \ \text{ iff } \ \ \forall (A,\alpha)\in \ModA(T) \ \ A,\alpha\models \phi \ \Longleftrightarrow \ A,\alpha\models \psi.
\]
The quotient $\FOA(T)=\FOA/{\approx}$, i.e.\ the Lindenbaum--Tarski algebra of $T$, carries a natural Boolean algebra structure. On the other hand, $\TypA(T)=\ModA/{\equiv}$ is naturally equipped with a topology, generated by the sets
\[
\sem{\phi}=\{[(A,\alpha)]\mid A,\alpha\models \phi\}
\]
for $\phi\in \FOA$, and is known as the \emph{space of types} of $T$. G{\"o}del's completeness theorem may now be stated as follows:
\begin{center}
the space $\TypA(T)$ is the Stone dual of $\FOA(T)$.
\end{center}
For every $n\in\N$, we can consider the Boolean subalgebra $\FO_n(T)$ of $\FOA(T)$ consisting of the equivalence classes of formulas with free variables in $v_1,\ldots,v_n$. The dual space of $\FO_n(T)$ is then the space of $n$-types of $T$. In particular, for $n=0$, we see that the dual space of the Lindenbaum-Tarski algebra of sentences $\FO_0(T)$ is the space of elementary equivalence classes of models of $T$.

Methods based on spaces of types play a central role in model theory. Their use can be traced back to Tarski's work, but the functorial nature of the construction was brought out and exploited nearly thirty years later by Morley in \citep{Morley1974}. In fact, it has been suggested that the notion of type space may be more fundamental than the notion of model \citep{Macintyre2003}. This point of view is related to the categorical approach, as the type space functor of a theory $T$ can be essentially identified with the (pointwise) dual of the hyperdoctrine associated with $T$.

This approach relies on the presentation of the algebra $\FOA(T)$ as the colimit of the following diagram of Boolean algebra embeddings:
\[\begin{tikzcd}
\FO_0(T) \arrow[hookrightarrow]{r} & \FO_1(T) \arrow[hookrightarrow]{r} & \FO_2(T) \arrow[hookrightarrow]{r} & \dots
\end{tikzcd}\]
Interestingly, this presentation does not fit with the inductive treatment of modal logic in Section~\ref{s:modal-Vietoris}, as the sentences, which is what we want to understand, belong to all the algebras in the chain. If we want to construct the Lindenbaum-Tarski algebra $\FOA(T)$ inductively, by adding a layer of quantifier $\exists$ at each step, we should start from the Boolean subalgebra $\FO^0(T)$ of $\FOA(T)$ consisting of the \emph{quantifier-free} formulas. The algebra $\FO^0(T)$ sits inside the algebra $\FO^1(T)$ of formulas with quantifier rank at most $1$, and so forth. The colimit of the diagram
\[\begin{tikzcd}
\FO^0(T) \arrow[hookrightarrow]{r} & \FO^1(T) \arrow[hookrightarrow]{r} & \FO^2(T) \arrow[hookrightarrow]{r} & \dots
\end{tikzcd}\]
is again the algebra $\FOA(T)$. In Section~\ref{s:Quant}, we will illustrate how the inductive methods used in B{\"u}chi's logic apply in the general first-order setting (and beyond) using the ideas set forth in Section~\ref{s:modal-Vietoris}.

\paragraph{B{\"u}chi's logic on words and profinite monoids.}
The connection between logic and automata goes back to the work of B{\"u}chi, Elgot, Rabin and others in the 1960s. In particular, B{\"u}chi's logic on words provides a powerful tool for the study of formal languages. The basic idea consists in regarding words on a finite alphabet $A$, i.e.\ elements of the free monoid $A^*$, as finite models for so-called \emph{logic on words}. 
That is, a word $w\in A^*$ is seen as a relational structure on the initial segment of the natural numbers \[\{1,\ldots, |w|\},\] where $|w|$ is the length of $w$, equipped with a unary relation $P_a$ for each $a\in A$ which singles out the positions in $w$ where the letter $a$ appears. 
B{\"u}chi's theorem states that the Lindenbaum-Tarski algebra of monadic second-order sentences for logic on words with the successor relation (interpreted over finite words) is isomorphic to the Boolean subalgebra of $\P(A^*)$ consisting of the regular languages \citep{Buchi1966}.

Since we are beyond first-order logic, and we have restricted to the finite models, the dual of the Lindenbaum-Tarski algebra is \emph{not} $A^*$, i.e.\ the collection of (elementary equivalence classes of) finite models. For the FO fragment of logic on words we can identify the dual with a space of models provided we allow for pseudofinite words. See e.g.\ \citep{vanGoolSteinberg}. However, this is not the case for monadic second-order logic and duality guides the right choice for the space of generalised models as the dual of the Lindenbaum-Tarski algebra. The latter coincides with (the underlying space of) the profinite completion $\w{A^*}$ of the monoid $A^*$, or equivalently, the free profinite monoid on the set $A$. 

The observation that the space underlying the free profinite monoid is the dual of the Boolean algebra of languages recognised by finite monoids essentially goes back to \citep{Birkhoff1937}, and was rediscovered by Almeida in the setting of automata theory \citep{Almeida1989}. Further, the fact that the monoid multiplication of $\w{A^*}$ also arises from duality for Boolean algebras with operators as the dual of certain quotienting operations on regular languages was shown in \citep{GGP2008}.

This type space tells us what generalised models for these logics should be, namely the points of the free profinite monoids. The realisation that these are an important tool in automata theory came in the 1980s \citep{Reiterman1982,Almeida94}. However, it was introduced, not via logic and duality, but rather via the connection between automata and finite semigroups, where the multiplication available on the profinite monoid also plays a fundamental role.

 An essential insight in the proof of B{\"u}chi's theorem is the fact that every monadic second-order formula is equivalent on words to an existential monadic second-order formula, and thus the iterative approach is not relevant as the hierarchy collapses. See \citep{GhivGo16} for a duality and type-theoretic approach via model companions. However, for the first-order fragment the iterative approach is very powerful. The first, and still prototypical application, is Sch{\"u}tzenberger's theorem which applies an iterative method, similar to the one of Section~\ref{s:modal-Vietoris}, to characterise the first-order fragment via duality. To be more precise, \citep{Schutzenberger65} shows that the star-free languages are precisely those recognised by (finite) aperiodic monoids. To prove this, Sch{\"u}tzenberger identified a semidirect product construction which captures dually the application of concatenation product on languages. The fact that star-free languages are precisely those given by first-order sentences of B{\"u}chi's logic was subsequently shown in \citep{MP1971}, though some passages in the introduction of \citep{Schutzenberger65} suggest that Sch{\"u}tzenberger was aware of this connection when he proved his result.

\paragraph{Domain Theory in Logical Form.}

In denotational semantics one seeks mathematical models of programs, which should be assigned in a compositional way. The compositionality means that program constructors should correspond to type constructors, and solutions to domain equations should correspond to program specifications. Scott's original solution to the domain equation 
\[
X\cong [X,X],
\]
seeking a domain $X$ which is isomorphic to the domain of its endomorphisms, was obtained by constructing a profinite poset, that is, a spectral space. Much further work confirmed that categorical methods, topology and in particular duality are central to the theory, cf.\ \citep{ScSt71,Plo76,SP82,Smyth83,LW91}. Rather than seeing Stone duality and its variants as useful technical tools for denotational semantics, Abramsky put Stone duality front and center stage: A \emph{program logic} is given in which denotational types correspond to theories and the ensuing Lindenbaum-Tarski algebras of the theories are bounded distributive lattices, whose dual spaces yield the domains as types. The constructors involved in the domain equations thus have duals under Stone duality, and solutions are obtained as duals of the solutions of the corresponding equation on the lattice side. In \citep{Abramsky87} Stone duality is restricted to the so-called Scott domains. That is, algebraic domains that are consistently complete. These are fairly simple and are closed under many constructors, including function space. In \citep{Abramsky91} the larger category of bifinite domains, which, in addition, is closed under powerdomain constructions, is used. We will say a bit more about bifinite domains later, but for now, we illustrate with a simple example at the level of spectral spaces.

\vspace{1em}
The Smyth powerdomain, $\mathbb S(X)$, is the space whose points are the compact and saturated\footnote{A subset $K\subseteq X$ is saturated provided it is an intersection of opens, or equivalently, it is an up-set in the specialisation order of the space $X$.} subsets of $X$ equipped with the upper Vietoris topology \citep{Smyth83}.  That is, the topology is generated by the subbasis given by the sets
\[
\Box U = \{K\in \mathbb S(X)\mid K\subseteq U\}, \ \ \text{for $U\subseteq X$ open}.
\]
At first sight, this may seem like quite an exotic object to pull out of a hat to study non-determinism. However, in Abramsky's duality with program logic, this construct is the Stone dual of adding a layer of (demonic) non-determinism. Indeed, if $X$ is a spectral space, then so is $\mathbb S(X)$, and if $L$ is the dual of $X$, then $\mathbb S(X)$ is the dual of
\[
F_\Box(L)=\mathbb F_{DL}(\Box L)/_{\approx}.
\]
Here, $\mathbb F_{DL}(\Box L)$ denotes the free distributive lattice\footnote{All distributive lattices are assumed to be bounded, and lattice homomorphisms preserve these bounds.} on the set of formal generators $\Box L=\{\Box a\mid a\in L\}$, and $\approx$ is the congruence given by the following scheme of relations between the generators:
\[
\Box(\bigwedge G)\ \approx \ \bigwedge\Box G \ \ \text{for $G\subseteq L$ finite}.
\]
Note that the Smyth powerdomain generalises the Vietoris hyperspace construction for Boolean spaces and, indeed, when $L=B$ is a Boolean algebra, the Booleanization of the lattice $F_\Box(B)$ coincides with the Boolean algebra $\MA(B)$ from Section~\ref{s:modal-Vietoris}.

Now the domain equation $X=\mathbb S(X)$ is solved by the final coalgebra for $\mathbb S$. However, a priori, there is no guarantee that it exists. On the other hand, the dual equation $L=F_\Box(L)$ is solved by the initial algebra, i.e.\ the free $\Box$-algebra over the empty set. As explained in Section~\ref{s:modal-Vietoris}, the latter algebra is guaranteed to exist since algebraic varieties are closed under filtered colimits.

Even though the duality theoretic paradigm supplied by the program logic makes it clearer why $\mathbb S(X)$ is the right object, one may still wonder how difficult it is to discover that $F_\Box(L)$ and $\mathbb S(X)$ are dual to each other. But this also is made quite algorithmic by duality: The dual of a free distributive lattice, such as $\mathbb F_{DL}(\Box L)$, is simply the Sierpinski cube $\two^{\Box L}$.\footnote{In this section, the dualizing object $\two$ is regarded as either a distributive lattice, or a spectral space by equipping the two-element set with the Sierpinski topology.} Indeed, a subset $S\subseteq \Box L$ corresponds to the unique homomorphism $h_S\colon \mathbb F_{DL}(\Box L)\to \two$ extending the characteristic map $\chi_S\colon \Box L\to \two$. Viewed as a theory (or prime filter) it is $F_S=\{\phi\mid \exists S'\subseteq S \text{ finite with }\bigwedge S'\leq\phi\}$. Also, a quotient of $\mathbb F_{DL}(\Box L)$ such as $F_\Box(L)$ is dual to a subspace of $\two^{\Box L}$, namely the one consisting of all those $S\subseteq \Box L$ such that 
\[
\Box(\bigwedge G)\in F_S\ \ \iff \ \  \bigwedge\Box G\in F_S, \ \ \text{ for $G\subseteq L$ finite}.
\]
By the definition of $F_S$, this is equivalent to 
\[
\Box(\bigwedge G)\in S \ \ \iff \ \ \Box G\subseteq S, \ \ \text{ for $G\subseteq L$ finite}.
\]
Note that $\two^{\Box L}$ is homeomorphic to $\P(L)$ with the topology generated by the sets $\widetilde{a}=\{S\in \P(L)\mid a\in S\}$ for $a\in L$. Viewed as subsets of $L$, the elements that belong to the dual of $F_\Box(L)$ are precisely the filters of $L$. 
That is, $\mathbb S(X)$ is homeomorphic to the space ${\rm Filt}(L)$ equipped with the topology generated by the sets $\widetilde{a}$ for $a\in L$. 
This algorithmic method, using duality for quotients of free algebras and then inductively adding layers of a connective, has been applied widely in the setting of propositional logics, see e.g.\ \citep{Ghilardi1992,BeGe10, Ghilardi10,CovG12}.\\

In \citep{Abramsky91} a large number of constructors such as $\mathbb S$ are treated, including the function space which, given two spaces $X$ and $Y$, yields the space $[X,Y]$ of all continuous functions $X\to Y$ in the compact-open topology. This case is more subtle, but it is closely related to the one above, and to the duality between lattices with residuation and Stone topological algebras, which is at the heart of the duality theory of profinite methods in automata theory. For these reasons, we go in a bit more detail. The following are extracts of a book in preparation \citep{GvGprepub}.

Consider the duality as above but for the operator type of implication. That is, given distributive lattices (DLs) $L$ and $M$, define
\[
F_\to(L\times M)=\mathbb F_{DL}(\to(L\times M))/_{\approx},
\]
where $\to(L\times M)=\{a\to b\mid a\in L, \, b\in M\}$ are the formal generators and $\approx$ is the congruence given by the following two schemes of relations between the generators:
\begin{enumerate}
\item[(i)] $a\to\bigwedge G=\bigwedge\{a\to b\mid b\in G\}$ for $a\in L$ and $G\subseteq M$ finite;
\item[(ii)] $\bigvee F\to b=\bigwedge\{a\to b\mid a\in F\}$ for $F\subseteq L$ finite and $b\in M$.
\end{enumerate}
Going through the same exercise as outlined above to identify the elements of $\two^{L\times M}$ which are compatible with the schemes (i) and (ii), one obtains the following result.

\begin{theorem}
Let $L$ and $M$ be DLs, and let $X$ and $Y$ be their respective dual spaces. The dual of $F_\to(L\times M)$ is the space $[X,\mathbb S(Y)]$ of continuous functions from $X$ to the Smyth powerspace of $Y$, in the compact-open topology.
\end{theorem}

This provides a dual description of $[X,\mathbb S(Y)]$, but we are interested in $[X,Y]$ which is a subspace of $[X,\mathbb S(Y)]$. However, it is not in general a closed subspace in the patch topology, reflecting the fact that $[X,Y]$ is not in general a spectral space. One would need to move to frames, sober spaces and geometric theories to describe $[X,Y]$ as the dual of a quotient. However, we have the following approximation.

\begin{proposition}\label{prop:preserves-joins-at-primes}
Let $L$ and $M$ be DLs, and $X,Y$ their respective dual spaces.  The dual of the quotient of $F_\to(L\times M)$ by a congruence $\theta$ is a subspace of $[X,Y]$ if and only if for all $x\in X$, $a\in F_x$, and finite subset $G\subseteq M$, there is $a'\in F_x$  such that
\[
[a\to(\bigvee G)]_\theta\ \leq\ [\bigvee\{a'\to b\mid b\in G\}]_\theta.
\]
Here, $F_x$ denotes the prime filter of $L$ corresponding to the point $x\in X$.
\end{proposition}

The above property may be thought of as saying that the operations $x\to(\ARG)$, for $x\in X$, preserve finite joins. For this reason, it has been called `preserving joins at primes'. Cf.\ Section~3.2 of \citep{Geh16}, where it is used to characterise the lattices with residuation that are dual to topological algebras based on Boolean spaces.

\vspace{1em}
There is a special case in which we can get our hands on the property of preserving joins at primes with a finitary scheme of relations between generators. This is the case where the lattice $L$ has enough join prime elements, i.e.\ every $a\in L$ is a finite join of join prime elements of $L$. This is for example true in free distributive lattices (where the meets of finite sets of generators are join prime), and it is intimately related to the interaction of domain theory and Stone duality as we have the following theorem.

\begin{theorem}\cite[Theorem 2.4.5]{Abramsky91}
A lattice has enough join primes if, and only if, its dual space endowed with the Scott topology is a domain. 
\end{theorem}

Let $L$ be a lattice with enough join primes, and $X$ its dual space. If $P=J(L)$ is the subposet of join prime elements of $L$, the free distributive lattice on the \emph{poset} $P$ is isomorphic to $L$. Further, $X\cong{\rm Idl}(P^{\op})$, the free directed join completion of $P^{\op}$ in the Scott topology, while $P^{\op}\cong{\rm Comp}(X)$, the set of compact elements of $X$. In particular, $X$ is an algebraic domain. Accordingly, we see that everything, i.e.\ $L$, $X$, and the compact elements of $X$, is determined by $P$. The posets $P$ that occur in this way were described already in \citep{Plo76}, where the profinite domains were characterised as those algebraic domains for which the set of compact elements form a `MUB-complete poset' in the nomenclature of \citep{AbrJung}. 
We now have a corollary of Proposition~\ref{prop:preserves-joins-at-primes}.

\begin{corollary}
Let $L$ and $M$ be DLs with dual spaces $X$ and $Y$, respectively. Suppose $L$ has enough join primes and let $P=J(L)$. Then the quotient of $F_\to(L\times M)$ by the congruence $\theta$ given by the following scheme is dual to the function space $[X,Y]$:
\[
p\to\bigvee G\approx\bigvee\{p\to b\mid b\in G\} \ \ \text{for $p\in P$ and $G\subseteq M$ finite}.
\]
\end{corollary}

In the above, we have just talked about spectral spaces and domains, but in order to have a class of spectral domains not only closed under function spaces and products, but also under the various versions of powerdomain, one must restrict oneself to the so-called bifinite domains. These were introduced (in the setting of domains with a least element) in \citep{Plo76} as generated by special MUB-complete posets $P$ now known as Plotkin orders \cite[Definition~4.2.1]{AbrJung}. These also have a beautiful very self-dual description relative to Stone duality.

The following definition applies to categories concrete over the category $\textbf{Pos}$ of posets and monotone maps, such as the category of DLs or that of spectral spaces and spectral maps (w.r.t.\ the specialization order) with the obvious forgetful functors.
\begin{definition}\label{def:embedding-retraction-pair}
Let $\mathcal C$ be a category equipped with a faithful functor $U\colon \mathcal{C}\to\mathbf{Pos}$. A pair of morphisms $C\xrightarrow{f}D\xrightarrow{g}C$ in $\mathcal C$ is an \emph{embedding-retraction-pair (e-r-p)} provided $(U(f),U(g))$ is an adjoint pair, and $U(f)$ is injective.\footnote{It follows from these two conditions that $U(f)$ is an embedding with left inverse $U(g)$.}
Further, such an e-r-p is said to be \emph{finite} if $U(C)$ is finite.\end{definition}

We have the following easy duality result.

\begin{proposition}\label{cor:e-r-p-duality}
In Stone duality, the dual of a (finite) embedding-retraction-pair on either side of the duality is a (finite) embedding-retraction-pair on the other side.
\end{proposition}

We may then define bifiniteness in the setting of spectral spaces, rather than in the setting of domains as it is customarily done.

\begin{definition}\label{def:spectral-bif}
Let $X$ be a spectral space, and $L$ its dual lattice. We say that $X$ and $L$ are \emph{bifinite} provided the following two equivalent conditions are satisfied:
\begin{enumerate}
\item $X$ is the cofiltered limit of the retractions of its finite e-r-p's;
\item $L$ is the filtered colimit of the embeddings of its finite e-r-p's.
\end{enumerate}
\end{definition}

The following proposition, which clearly implies that a bifinite lattice must have enough join primes, allows us to conclude that bifinite spectral spaces are bifinite domains. Thus, the above definition is no more general than the standard one.

\begin{proposition}\label{prop:fin-lat-e-r-p}
Let $L$ be a distributive lattice and $K\subseteq L$ a finite sublattice. Then the following conditions are equivalent:
\begin{enumerate}
\item There is a lattice homomorphism $h\colon L\to K$ making $(i,h)$ an embedding-retraction-pair, where $i\colon K\to L$ is the inclusion;
\item
\begin{enumerate}
\item[(i)] For all $b\in L$,  ${\downarrow}b\cap K$ is a principal downset;
\item[(ii)] $J(K)\subseteq J(L)$.
\end{enumerate}
\end{enumerate}
\end{proposition}

\section{Quantifiers, free constructions and duality}\label{s:Quant}

In the categorical logic approach, cf.\ Sections~\ref{s:algebras-from-logic} and~\ref{s:three-ex-spaces}, the stratification of the algebra of formulas (up to logical equivalence modulo $T$) provided by the hyperdoctrine $P\colon \Con^\op\to\BA$ is in a sense impredicative. Indeed, it starts from the algebra of sentences $P(\emptyset)$, which is what we ultimately want to understand, to build all formulas on a countably infinite set of variables. 
This contrasts with the step-wise construction of algebras of formulas outlined in Section~\ref{s:modal-Vietoris}.

We want to understand quantification as a step-by-step construction. To this end, in this section we analyse from a duality theoretic viewpoint the inductive process of applying a layer of quantifiers in three settings. First, we focus on existential quantification in first-order logic over arbitrary structures. Then, on semiring and probabilistic quantifiers in first-order logic over finite structures. 

As explained in Section~\ref{s:algebras-from-logic}, Lindenbaum-Tarski algebras of predicate logics typically fail to be free algebras. The challenge then consists, in a sense, in building free objects which approximate the Lindenbaum-Tarski algebra we are interested in. We illustrate this idea in the following examples.

\subsection{Existential quantification and Vietoris}\label{s:exists-vietoris}
For existential quantification in first-order logic, the framework can be loosely described as follows. Assume we are given a Boolean algebra of formulas $B$, and we build a new Boolean algebra $B_{\exists x}$ by adding a layer of the quantifier $\exists x$ to the formulas in $B$. We then have a quotient map
\[\begin{tikzcd}
\MA(B) \arrow[twoheadrightarrow]{r} & {B_{\exists x}}
\end{tikzcd}\]
sending $\Diamond \phi$ to $\exists x.\phi$, where $\MA(B)$ is the Boolean algebra obtained by freely adding one layer of modality as described in Section~\ref{s:modal-Vietoris}. Dually, we get a continuous embedding
\[\begin{tikzcd}
\V(X) \arrow[hookleftarrow]{r} & {X_{\exists x}}
\end{tikzcd}\]
where $X$ and $X_{\exists x}$ are the dual spaces of $B$ and $B_{\exists x}$, respectively. We have approximated the space $B_{\exists x}$ by means of the Vietoris space $\V(X)$, whose dual is a \emph{free} object (namely, the free modal algebra on $B$). The problem then consists in characterising $X_{\exists x}$ as a subspace of $\V(X)$. This is addressed by observing that $X_{\exists x}$ is the image of a continuous map into $\V(X)$ constructed in a canonical way. In the remaining of this section we provide the necessary details.

Recall from Section~\ref{s:three-ex-spaces} that a first-order formula $\phi \in \FOA(T)$ can be identified with the set $\sem{\phi}\subseteq \ModA/{\equiv}$ consisting of the (equivalence classes of) models with assignments satisfying $\phi$. If the free variables of $\phi$ are contained in $v_1,\dots,v_n$, we can restrict the variable assignments accordingly. Write
\[
\Mod_n=\{[(A,\alpha\colon \{v_1,\ldots,v_n\}\to A)]\mid A \ \text{is a $\Sg$-structure and} \ A\models T\},
\]
where $[(A,\alpha)]=[(A',\alpha')]$ if and only if $A,\alpha\models \phi \Leftrightarrow A',\alpha'\models \phi$ for every $\phi\in \FO_n(T)$. Henceforth, we abuse notation and denote an arbitrary element of $\Mod_n$ by $(A,\alpha)$ instead of $[(A,\alpha)]$. Then, $\FO_n(T)$ embeds into $\P(\Mod_n)$ via the map
\begin{align*}
\FO_n(T) \hookrightarrow \P(\Mod_n), \ \ [\phi]\mapsto \sem{\phi}_n=\{(A,\alpha)\in \Mod_n \mid A,\alpha\models \phi\}.
\end{align*}
The projection map
\[
\pi_i\colon \Mod_n\epi \Mod_{n\setminus i}
\]
which forgets the value of the assignments on the variable $v_i$ induces a Boolean algebra embedding 
\[
\pi_i^{-1}\colon \P(\Mod_{n\setminus i}) \mono \P(\Mod_n)
\] 
by applying the contravariant power-set functor. As in the hyperdoctrine approach, the homomorphism $\pi_i^{-1}$ has a lower adjoint and it is given by taking direct images under $\pi_i$.
\[\begin{tikzcd}[column sep=2.0em]
\P(\Mod_{n\setminus i}) \arrow[bend left=35, looseness=1]{rr}[description]{\pi_i^{-1}} & {\footnotesize{\text{$\top$}}} &  \P(\Mod_n) \arrow[bend left=35, looseness=1]{ll}[description]{\pi_i(\ARG)}
\end{tikzcd}\]
This lower adjoint map can be thought of as the quantifier $\exists v_i$. Indeed, it is readily seen that $\pi_i(\sem{\phi}_n)=\sem{\exists v_i. \phi}_{n\setminus i}$.
More generally, abstracting away from the Boolean subalgebra $\FO_n(T) \mono \P(\Mod_n)$, we can consider any Boolean algebra embedding
\[
j\colon B\hookrightarrow \P(\Mod_n)
\]
and regard it as a `semantically given logic'. The Boolean algebra obtained by adding a layer of the quantifier $\exists v_i$ to $B$ can be identified with the Boolean subalgebra $B_{\exists}^i$ of $\P(\Mod_{n\setminus i})$ generated by the set of direct images
\[
\{\pi_i(j(\phi))\mid \phi\in B\}.
\]

We now focus on the dual of the transformation $B\leadsto B_{\exists}^i$. Let $f\colon \beta(\Mod_n)\epi X$ be the continuous map dual to $j\colon B\hookrightarrow \P(\Mod_n)$. Here, $\beta(\Mod_n)$ denotes the \v{C}ech-Stone compactification of $\Mod_n$ regarded as a discrete space, and is the dual Stone space of $\P(\Mod_n)$. We obtain a continuous map
\[\begin{tikzcd}[column sep=3.5em]
R\colon \beta(\Mod_{n\setminus i}) \arrow{r}{\beta(\pi_i)^{-1}} & \V(\beta(\Mod_n)) \arrow{r}{\V(f)} & \V(X). 
\end{tikzcd}\]
The first component of $R$ is the preimage map $x\mapsto \beta(\pi_i)\inv(x)$, where the function $\beta(\pi_i)\colon \beta(\Mod_n) \to \beta(\Mod_{n\setminus i})$ is the Stone dual of $\pi_i^{-1}\colon \P(\Mod_{n\setminus i}) \to \P(\Mod_n)$. The map $\beta(\pi_i)^{-1}$ is continuous because $\pi_i^{-1}$ has a lower adjoint. Indeed, the join-semilattice homomorphism $\pi_i(\ARG)\colon \P(\Mod_n)\to \P(\Mod_{n\setminus i})$ induces a Boolean algebra homomorphism $\MA(\P(\Mod_n))\to \P(\Mod_{n\setminus i})$, whose dual map is precisely $\beta(\pi_i)^{-1}$.

We then have the following result.
\begin{proposition}\label{p:exists-Vietoris}
The image of the continuous map $R\colon \beta(\Mod_{n\setminus i}) \to \V(X)$ is the dual space of $B_{\exists}^i$.
\end{proposition}
\begin{proof}
 It is not difficult to verify that $R^{-1}(\Diamond \w{\phi})=\w{\pi_i(j(\phi))}$ for every $\phi\in B$, see e.g.\ Corollary 3.2 of \citep{BG2019}. Consequently, the Boolean algebra dual to the image of $R$ can be identified with the subalgebra of $\P(\Mod_{n\setminus i})$ generated by the elements of the form $\pi_i(j(\phi))$ for $\phi\in B$, which is precisely $B_{\exists}^i$.
\end{proof}

To sum up, the transformation $B\leadsto B_{\exists}^i$ which adds one layer of quantifier $\exists v_i$ dually corresponds to taking the image of the continuous map $R\colon \beta(\Mod_{n\setminus i})\to \V(X)$, canonically constructed from the continuous function $f\colon \beta(\Mod_n)\epi X$. For a step-by-step treatment of quantifiers, we now want to add to $B_{\exists}^i$ the formulas which were already in $B$. Hence, we take the Boolean subalgebra of $\P(\Mod_n)$ generated by the union $B\cup B_{\exists}^i$, which coincides with the image of the obvious Boolean algebra homomorphism $B+B_{\exists}^i\to \P(\Mod_n)$. This corresponds, dually, to taking the image of the continuous  product map
\[\begin{tikzcd}[column sep=5.5em]
\beta(\Mod_n) \arrow{r}{(R\circ \beta(\pi_i))\times f} & \V(X)\times X.
\end{tikzcd}\]
An essential obstacle to a two-sided duality theory for quantifiers is the lack of a characterisation of the continuous maps $\beta(\Mod_{n})\to \V(X)\times X$ arising this way. We will return to this point in Section~\ref{s:outlook}.

\subsection{Semiring quantifiers and measures}\label{s:semiring-quant}
The existential quantifier $\exists$ captures the existence, or non-existence, of an element satisfying a property. As such, it is a two-valued query. Semiring quantifiers, as studied for instance in logic on words, generalise $\exists$ by allowing us to count the number of witnesses in a given semiring.\footnote{A particular class of semiring quantifiers is given by the modular quantifiers, which count in a finite cyclic ring $\mathbb{Z}/q\mathbb{Z}$. These were introduced in logic on words in \citep{STRAUBINGTT}.} Recall that a \emph{semiring} is a tuple $(S,+,\cdot,0,1)$ where $(S,+,0)$ is a commutative monoid, $(S,\cdot,1)$ is a monoid, the operation $\cdot$ distributes over $+$, and $0\cdot s=0=s\cdot 0$ for all $s\in S$.
If $S$ is a fixed finite semiring, every element $k\in S$ determines a quantifier $\exists_{k}$. Given a first-order formula $\phi$ with one free variable $v$ and a finite structure $A$, the semantics of the sentence $\exists_{k}v.\phi(v)$ is given as follows:
\begin{align*}
    A \models \exists_{k}v.\phi(v)
    & \qtq{iff} 1+\cdots+1 \text{ (repeated $m$-times) is equal to $k$ in $S$} \\
    & \text{where $m$ is the number of elements $a\in A$ such that $A\models \phi(a)$.}
\end{align*}
Notice that $A$ must be finite, for otherwise the set $\{a\in A\mid A\models \phi(a)\}$ may be infinite and the sum $1+\cdots+1$ undefined. This problem could be overcome by requiring that $S$ be complete in an appropriate sense. The existential quantifier $\exists$ is recovered by letting $S=\two$ be the two-element Boolean ring and $k=1$.

Let $\Fin_n$ be the subset of $\Mod_n$ consisting of the finite models with assignments. Given a Boolean algebra embedding $j\colon B\hookrightarrow \P(\Fin_n)$ we can construct, akin to the case of $\exists$, a Boolean algebra $B_{\exists_S}^i$ obtained by adding a layer of semiring quantifiers $\exists_{k}v_i$ for $k\in S$. For every $\phi\in B$ and $(A,\alpha)\in \Fin_{n\setminus i}$, write $m_{\phi,(A,\alpha)}$ for the number of elements $a$ in $A$ such that $(A, \alpha\cup \{v_i \mapsto a\})$ belongs to $j(\phi)$. Then, $B_{\exists_S}^i$ can be defined as the Boolean subalgebra of $\P(\Fin_{n\setminus i})$ generated by the sets
\[
\{(A,\alpha) \in \Fin_{n\setminus i} \mid 1 + \cdots + 1 \text{ ($m_{\phi,(A,\alpha)}$-times) is equal to $k$} \}, \ \text{for} \ \phi\in B \ \text{and} \ k\in S.
\]

In order to describe the dual of the transformation $B\leadsto B_{\exists_S}^i$, we need to understand which construction plays the role of the Vietoris hyperspace in the case of semiring quantifiers. For this purpose, notice that the Vietoris space $\V(X)$ can be identified with a space of two-valued finitely additive measures on $X$, whenever $X$ is a Boolean space.\footnote{Perhaps more natural would be to first identify $\V(X)$ with the space of filters on the dual Boolean algebra of $X$, as explained towards the end of Section~\ref{s:three-ex-spaces} in the case of the Smyth powerspace, and then observe that filters can be seen as two-valued finitely additive measures.} Regard $X$ as a measurable space where the measurable subsets are precisely the clopens, i.e.\ the elements of the Boolean algebra $B$ dual to $X$. A finitely additive $\two$-valued measure on $X$ is then a function $\mu\colon B\to \two$ satisfying
\[
\mu(0)=0 \qtq{and} \mu(a\vee b)\vee \mu(a\wedge b)=\mu(a)\vee \mu(b)  \ \ \forall a,b\in B.
\]
Denote by $\M(X,\two)$ the collection of all finitely additive $\two$-valued measures on $X$, and equip it with the subspace topology induced by the product topology on $\two^B$.
\begin{proposition}
For every Boolean space $X$, the Vietoris hyperspace $\V(X)$ is homeomorphic to $\M(X,\two)$ via the map
\[
\V(X)\to \M(X,\two), \ \ C\mapsto \mu_C, \ \ \text{where} \ \ \mu_C(a)=\begin{cases} 1 & \mbox{if \ $\w{a}\cap C\neq \emptyset$,} \\ 0 & \mbox{otherwise.} \end{cases}
\]
\end{proposition}
\begin{proof}
It is straightforward to verify that the map in the statement is a continuous bijection, with inverse
$\M(X,\two)\to\V(X)$, $\mu\mapsto \bigcap{\{\w{a}\subseteq X\mid \mu(\neg a)=0\}}$.
Every continuous bijection between compact Hausdorff spaces is a homeomorphism, hence the statement follows. 
\end{proof}
For semiring quantifiers, the hyperspace $\V(X)$ will thus be replaced by $\M(X,S)$, the space of finitely additive $S$-valued measures on $X$. An element of $\M(X,S)$ is a function $\mu\colon B\to S$ satisfying
\begin{equation}\label{eq:fin-add}
\mu(0)=0 \qtq{and} \mu(a\vee b)+ \mu(a\wedge b)=\mu(a)+ \mu(b)  \ \ \forall a,b\in B,
\end{equation}
and the set $\M(X,S)$ is equipped with the subspace topology induced by the product topology on $S^B$. The equations in~\eqref{eq:fin-add}, encoding finite additivity, translate into equaliser diagrams in the category of Boolean spaces. Hence, the resulting space $\M(X,S)$ is again Boolean. Explicitly, the topology of $\M(X,S)$ is generated by the (clopen) subsets of the form
\[
[a,k]=\{\mu\in\M(X,S)\mid \mu(a)=k\}, \ \ \text{for} \ a\in B \ \text{and} \ k\in S.
\]

In order to describe the dual of the construction $B\leadsto B_{\exists_S}^i$, we perform two steps. First, given a finite model with assignment $(A,\alpha)\in \Fin_{n\setminus i}$, let
\begin{equation}\label{eq:fsp-semiring}
\indi_{(A,\alpha)}\colon \Fin_n \to S
\end{equation}
be the `$S$-valued characteristic function' of $\pi_i^{-1}(A,\alpha)$,  
where $\pi_i\colon \Fin_n \to \Fin_{n\setminus i}$ is the map which forgets the assignment of the $i$th variable. That is, $\indi^i_{(A,\alpha)}(A',\alpha')$ is $1$ if $A = A'$ and $\alpha$ agrees with $\alpha'$ on the variables $v_1, \dots, v_{i-1}, v_{i+1}, \dots, v_n$, and $0$ otherwise. Since $A$ is finite, $\indi^i_{(A,\alpha)}$ belongs to the set $\S(\Fin_n)$ of finitely supported $S$-valued functions on $\Fin_n$.
In the second step, in order to construct a measure, we extend the function $\indi^i_{(A,\alpha)}$ to subsets of $\Fin_n$ by adding up all the non-zero values in a given subset. More generally, if  $T$ is a set and $g\colon T\to S$ is a finitely supported function, the map
\[ \int g\colon \P(T) \to S,\quad P \mapsto \int_P g \qtq{computed as} \sum_{x\in P} g(x) \]
is a finitely additive $S$-valued measure on $\beta(T)$. We obtain an integration map\footnote{In fact, the construction $X\mapsto \M(X,S)$ yields a monad on $\BStone$ and the integration map can be upgraded to a monad morphism $\int\colon \S\circ U\to U\circ \M(-,S)$, where $\S$ is the semiring monad on $\Set$ and $U\colon \BStone\to \Set$ is the forgetful functor. Cf.\ \citep{GPR2017}. While, for the purpose of this section, we may assume $S$ is any pointed monoid, the monadic treatment requires the full semiring structure.}
\[
\int\colon \S(T) \to \M(\beta(T), S).
\]

Now, let $f\colon \beta(\Fin_n) \to X$ be the dual of the embedding $j\colon B \mono \P(\Fin_n)$. 
Consider the composite
\begin{equation}\label{eq:Fin-to-meas-semirings}
    \Fin_{n\setminus i} \xrightarrow{\ee{\indi^i_{(\ARG)}}} \S(\Fin_n) \xrightarrow{\ee\int} \M(\beta(\Fin_n), S) \xrightarrow{\ee{f_*}} \M(X,S)
\end{equation}
where $f_*$ sends a measure to its pushforward along $f$, i.e.\ $f_*(\mu)(a)=\mu(f^{-1}(\w{a}))$ for every $\mu\in \M(\beta(\Fin_n),S)$ and $a\in B$.
The space $\M(X,S)$ is compact and Hausdorff, whence the above composition extends to a (unique) continuous function
\begin{equation}\label{eq:map-R}
R \colon \beta(\Fin_{n\setminus i}) \to \M(X, S).
\end{equation}
The following result generalises Proposition~\ref{p:exists-Vietoris} and can be proved in a similar manner (we omit the details here).
\begin{theorem}\label{t:semiring-quant-measures}
The image of the continuous map $R\colon \beta(\Mod_{n\setminus i}) \to \M(X,S)$ is the dual space of $B_{\exists_S}^i$.
\end{theorem}
The connection between semiring quantifiers and spaces of finitely additive measures was first explored, in the context of logic on words, in \citep{GPR2017}. The treatment in this section could be adapted to deal with any profinite semiring, such as the \emph{tropical semiring} $(\mathbb{N}\cup\{\infty\},\min,+,\infty,0)$, and not just the finite ones. See \citep{R2020}.

\subsection{Probabilistic quantifiers and structural limits}
Topological methods are also employed in the study of structural limits in finite model theory. A systematic investigation of limits of finite structures has been developed by Ne{\v s}et{\v r}il and Ossona de Mendez and is based on an embedding, called the \emph{Stone pairing}, of the collection of finite structures into a space of probability measures \citep{NO2012,NOdM2020}. The latter space is complete, thus it provides the limit objects for those sequences of finite structures which embed as Cauchy sequences. Although this space of measures and the Stone pairing embedding did not originate from duality, in recent work we showed that a closely related version of the Stone pairing can be understood  --- via duality --- as the embedding of finite structures into a space of types. Namely, the space of $0$-types of an extension of first-order logic obtained by adding a layer of certain probabilistic quantifiers \citep{GJR2020}. In the following, we highlight the similarities between the Stone pairing embedding and the space-of-measures construction introduced above in the context of existential and semiring quantification.

For every first-order formula $\phi$ with free variables contained in $v_1,\dots,v_n$, and finite structure $A$, the \emph{Stone pairing} of $\phi$ and $A$ is defined as
\[ \SP{\phi,A} \ = \ \frac{|\{ \overline a \in A^n \mid A \models \phi(\overline a)\}|}{|A|^n}. \]
In other words, $\SP{\phi,A}$ is the probability that a random assignment of the variables $v_1,\dots,v_n$ in $A$ satisfies the formula $\phi$.
Upon fixing the second coordinate, the map $\SP{\ARG,A}$ is a finitely additive measure on the dual space of the Lindenbaum-Tarski algebra of all first-order formulas $\FOA$, with values in the unit interval $[0,1]$. I.e.,
\[
\SP{\bot,A}=0 \qtq{and} \SP{\phi\vee \psi,A}+ \SP{\phi\wedge \psi,A}=\SP{\phi,A}+ \SP{\psi,A}  \ \ \forall \phi,\psi \in\FOA.
\]
Since the Boolean algebra $\FOA$ is dual to the space of models and valuations $\ModA$, we obtain an embedding
\begin{align*}
    \SP{\ARG,\ARG}\colon \Fin \longrightarrow \M(\ModA,[0,1]),\quad A \mapsto \SP{\ARG,A}
\end{align*}
where $\Fin$ is the collection of finite structures, up to isomorphism (with the notation of Section~\ref{s:semiring-quant}, $\Fin=\Fin_0$). This is the \emph{Stone pairing} embedding introduced by \NOdM{}. 

By restricting $\SP{\ARG,A}$ to suitable fragments of first-order logic, \NOdM{} obtained a unifying framework that captures various notions of convergence of finite structures, such as Lovasz--Szegedy convergence, Benjamini--Schramm convergence, elementary convergence, etc.\footnote{Note that the restriction of the Stone pairing embedding to a fragment of FO may fail to be injective.} Their insight was that each of these notions of convergence corresponds to a fragment of first-order logic. Further, since the ensuing spaces of finitely additive measures are complete, they admit a limit for every sequence of finite structures which embeds as a Cauchy sequence. 

In section~\ref{s:semiring-quant}, we defined a map from a set of finite structures with evaluations into a space of finitely additive measures, see equation~\eqref{eq:Fin-to-meas-semirings}, and showed that it dually captures the adding of a layer of semiring quantifiers. By analogy, we may ask if the Stone pairing also corresponds to applying a layer of quantifiers. One immediate obstacle is that the spaces $[0,1]$ and $\M(\ModA,[0,1])$ are not Boolean, whence not amenable to the methods of Stone duality for Boolean algebras.

We can overcome this problem by replacing $[0,1]$ with a profinite version of the unit interval obtained from a codirected system of finitary approximations of real numbers in $[0,1]$. This profinite space $\G$ is naturally equipped with a Priestley space structure and can therefore be studied using Stone-Priestley duality for distributive lattices.  To define $\G$, we divide the unit interval into $n$ segments of equal length, i.e.
\[ \G_n \ee= \{ 0 \ee< \tfrac{1}{n} \ee< \tfrac{2}{n} \ee< \dots \ee< 1\}. \]
The chain $\G_n$ provides a finite approximation of $[0,1]$. The higher the value of $n\in\N$, the better the approximation is. Whenever $n\mid m$, we consider the flooring function $\G_m\to \G_n$ sending $\frac{a}{m}$ to the largest $\frac{b}{n} \in \G_n$ such that $\frac{b}{n} \leq \frac{a}{m}$. Note that the finite chains $\G_n$ with flooring functions between them form a codirected diagram in the category $\Posf$ of finite posets with monotone maps. The limit of this diagram is an object $\G$ of the pro-completion of $\Posf$, which is the category of Priestley spaces with continuous monotone maps.\footnote{A \emph{Priestley space} is a pair $(X,\leq)$ where $X$ is a compact space and $\leq$ is a partial order such that, whenever $x\not\leq y$, there is a clopen subset $C\subseteq X$ which is upward closed and satisfies $x\in C$ and $y\notin C$.} See e.g.\ Corollary~VI.3.3 in \citep{Johnstone1986}.

Concretely, the elements of $\G$ are the sequences of approximations $(x_n)_n\in\prod_{n\in\N}{\G_n}$ which are compatible with the flooring functions. Every $q\in (0,1]$ determines an element $q\mm\in\G$, namely the sequence
\[
q\mm=(q_1\mm,q_2\mm,q_3\mm,\ldots) \qtq{where} q_n\mm=\max \{ \tfrac{a}{n} \in \G_n \mid \tfrac{a}{n} < q \}
\]
which approximates $q$ from below while never reaching it. Further, if $q$ is rational, we also get a lower approximating sequence $q\cc\in\G$ which eventually stabilises at $q$:
\[
q\cc=(q_1\cc,q_2\cc,q_3\cc,\ldots) \qtq{where} q_n\cc=\max \{ \tfrac{a}{n} \in \G_n \mid \tfrac{a}{n} \leq q \}.
\]
In fact, any point of $\G$ is of one of these two types. We can thus think of $\G$ as a copy of the unit interval where all the non-zero rationals are doubled (in the picture, $q$ is rational while $r$ is irrational):
\begin{center}
\begin{tikzpicture}
    \node at (6.15,0) (1cc) {};
    \node at (6,0) (1mm) {};
    \node at (0,0) (0cc) {};
    \node at (4.4,0) (r) {};
    \node at (1.75,0) (qc) {};
    \node at (1.60,0) (qm) {};
    \node at ($(r) -(-0.1,0.5)$) (rmm) {$r\mm$};
    \node at ($(qc)+(0.15,0.5)$) (qcc) {$q\cc$};
    \node at ($(qm)-(0.05,0.5)$) (qmm) {$q\mm$};
    \node at ($(1cc)+(0.1,0.5)$) {$1\cc$};
    \node at ($(1mm)-(0.05,0.5)$) {$1\mm$};
    \node at ($(0cc)+(0.15,0.5)$) {$0\cc$};

    \draw[densely dotted] (1mm.center) -- (qc.center);
    \draw[densely dotted] (qm.center) -- (0cc.center);

    \foreach \pt in {1cc,1mm,0cc,r,qm,qc} {
        \draw ($(\pt)-(0,0.1)$) -- ($(\pt)+(0,0.1)$);
    }
    \node at (-1.1, 0) {$\G\enspace =$};
\end{tikzpicture}
\end{center}

Equivalently, $\G$ is a copy of the Cantor space with an extra top element which is topologically isolated (corresponding to $1\cc$). The natural order of $\G$, illustrated in the previous picture, is the total order defined by the two conditions
\begin{itemize}
    \item $r\cc < s\mm$ if and only if $r < s$ in $[0,1]$, and
    \item $q\mm < q\cc$ for every $q\in (0,1]$,
\end{itemize}
and its topology is the interval topology. Note that $\G$ retracts onto $[0,1]$. Indeed, the continuous surjection 
\[
\gamma\colon \G\to[0,1], \ \ q\mm,q\cc\mapsto q
\] 
has a (lower semicontinuous) section
\[
\iota\colon [0,1]\to \G, \ \ \iota(q)=\begin{cases}
       q\cc & \text{if $q$ is rational} \\
       q\mm & \text{otherwise}.
   \end{cases}
\]

The additive structure of $[0,1]$ lifts to $\G$ (as can be derived by duality for additional operators) so that it makes sense to consider the set $\M(X,\G)$ of finitely additive probability measures on a Boolean space $X$ with values in $\G$. This construction can be generalised to any Priestley space $X$, and it turns out that the assignment $X\mapsto \M(X,\G)$ is an endofunctor on the category of Priestley spaces. In particular, a continuous monotone map of Priestley spaces $f\colon X\to Y$ is sent to the map
\[
f_*\colon \M(X,\G)\to \M(Y,\G)
\]
taking a measure to its pushforward along $f$. Furthermore, the retraction-section pair $\gamma\colon \G\leftrightarrows [0,1]\cocolon \iota$ lifts to a retraction-section pair 
\[
\gamma^\#\colon \M(X,\G)\leftrightarrows \M(X,[0,1])\cocolon \iota^\#, \qtq{where} \gamma^\#(\mu)=\gamma \circ \mu \qtq{and} \iota^\#(\mu)=\iota\circ\mu.
\]

Now we define a $\G$-valued variant of the Stone pairing by following the strategy set out in Section~\ref{s:semiring-quant} in the case of semiring quantifiers.
Fix $n\in\N$, and let $\F(\Fin_n, \G)$ be the set of finitely supported functions $\Fin_n\to\G$ with total value $1\cc$. We get a map $\delta_{(\ARG)}\colon \Fin \to \F(\Fin_n, \G)$ sending a finite structure $A$ to 
\[ \delta_A\colon \Fin_n \to \G, \qtq{where} \delta_A(A',\alpha') =
   \begin{cases}
       \left(\frac{1}{|A|^n}\right)\cc & \text{if } A' = A \\[0.7em]
       0\cc & \text{otherwise}.
   \end{cases}
\]
The map $\delta_{(\ARG)}$ is the (normalized) $\G$-valued version of the function introduced in~\eqref{eq:fsp-semiring} for semiring quantifiers. In a similar way, to move from finitely supported functions to measures, for every set $T$ we consider the integration map
\[ 
\int\colon \F(T,\G) \to \M(\beta(T),\G),\quad f\mapsto \int f. 
\]
Lastly, define the following composition
\[\begin{tikzcd}
R_n\colon \Fin \arrow{r}{\delta_{(\ARG)}} & \F(\Fin_n,\G) \arrow{r}{\int} & \M(\beta(\Fin_n),\G) \arrow{r}{f_*} & \M(\Mod_n, \G) 
\end{tikzcd}\]
where $f\colon \beta(\Fin_n) \to \Mod_n$ is the dual map of the Boolean algebra homomorphism 
\[
\FO_n \to \P(\Fin_n), \ \ \phi \mapsto \sem{\phi}\cap \Fin_n.
\] 
The map $R_n$ can be extended to a continuous function $\widetilde{R}_n\colon \beta(\Fin)\to \M(\Mod_n,\G)$, corresponding to the map in~\eqref{eq:map-R}. Using the fact that the space $\ModA$ is the codirected limit of the $\Mod_n$'s for $n\in\N$, and the functor $\M(\ARG,\G)$ preserves codirected limits, we can `glue' the maps $\widetilde{R}_n$ to get a continuous function $\widetilde{R}\colon \beta(\Fin) \to \M(\ModA,\G)$. The restriction $R\colon \Fin\to\M(\ModA,\G)$ of $\widetilde{R}$ is an equivalent $\G$-valued version of the Stone pairing, as expressed by the commutativity of the following diagram.
\[\begin{tikzcd}[row sep=2em]
{} & \M(\ModA,\G)\ar[bend left=25]{dd}{\gamma^\#} \\
\Fin \ar{ru}{R}\ar[swap]{rd}{\SP{\ARG,\ARG}} & \\
& \M(\ModA,[0,1])\ar[bend left=25]{uu}{\iota^\#}
\end{tikzcd}\]
The map $R$, and more precisely the way it is constructed, provides an interesting link between the theory of structural limits and the inductive study of semiring quantifiers. Further, the duality approach allows us to see (the $\G$-valued version of) the Stone pairing as an embedding of the finite structures into a space of types. This is the content of the following theorem, which is a special case of more general results in \citep{GJR2020}.
\begin{theorem}
The Boolean space $\M(\ModA,\G)$ is dual to the Lindenbaum-Tarski algebra of the propositional logic having as atoms $\PrG q \phi$ and $\PrL q \phi$, for each $\phi \in \FOA$ and $q\in [0,1]\cap \mathbb Q$, and the following inference rules (along with the usual ones for the Boolean connectives):
\begin{equation*}
   \addtolength{\fboxsep}{4pt}
    \boxed{
\begin{gathered}
\infer[{\scriptstyle(\mathrm{if} \ p\, \leq \, q)}]
{\PrG p \phi}{\PrG q \phi}
\hspace{1.5em}
\infer[{\scriptstyle(\mathrm{if} \ \phi \, \vdash \, \psi)}]
{\PrG q \psi}{\PrG q \phi}
\hspace{1.5em}
\infer
{\PrG 0 \bot}{}
\hspace{1.5em}
\infer[{\scriptstyle(\mathrm{if} \ q \, > \, 0)}]
{\PrL q \bot}{}
\hspace{1.5em}
\infer
{\PrG q \top}{}
\hspace{1.5em}
\infer=
{\neg \PrL q \phi}{\PrG q \phi} \\[2ex]
\infer
{\PrG{p+q-r}(\phi\vee \psi) \,\vee\, \PrG{r}(\phi\wedge \psi)}{\PrG p \phi\,\wedge\, \PrG q \psi}
\hspace{1em}
\infer[{\hspace{0.4em}\scriptstyle(\mathrm{if} \ 0 \, \leq \, p+q-r \, \leq \, 1)}]
{\PrG p \phi \,\vee\, \PrG q \psi}{\PrG{p+q-r}(\phi\vee \psi)\,\wedge\,\PrG{r}(\phi\wedge \psi)}
\end{gathered}
}\end{equation*}
\end{theorem}
The intended models for this extension of FO are the measures $\mu\in\M(\ModA,\G)$, and the probabilistic quantifiers $\PrG q$ and $\PrL q$ are interpreted as follows:
\[
\mu\models \PrG q \phi \ \Leftrightarrow \ \mu(\phi)\geq q\cc \qtq{and} \mu\models \PrL q \phi \ \Leftrightarrow \ \mu(\phi)< q\cc.
\]
In particular, if $A$ is a finite structure, $\SP{\ARG,A}\models \PrG q \phi$ if and only if $\phi$ is satisfied in $A$ with probability at least $q$. Similarly for $\PrL q \phi$. Note that these probabilistic quantifiers bind all free variables in a formula. Thus, once applied a layer of quantifiers to $\FOA$, we obtain an algebra of \emph{sentences}. These sentences are seen as propositional atoms for a new logic and, by the previous theorem, the Stone pairing can be seen as embedding the collection of finite structures (up to isomorphism) into the space of $0$-types for this logic.

Therefore, we see that \NOdM's Stone pairing dually corresponds to adding a layer of probabilistic quantifiers. As such, it can be regarded as an instance of the inductive approach described in Section~\ref{s:modal-Vietoris}.

\section{Outlook}\label{s:outlook}
We saw in Section~\ref{s:exists-vietoris} that adding a layer of existential quantifier $\exists$ to a Boolean algebra $B$ of first-order formulas (with free variables in $v_1,\ldots,v_n$) dually corresponds to taking the image of a continuous map $\beta(\Mod_{n})\to \V(X)\times X$, where $X$ is the dual Stone space of $B$. A similar statement holds for semiring quantifiers, cf.\ Section~\ref{s:semiring-quant}. This continuous map is defined in a canonical way, and ensures the \emph{soundness} of the construction. But we do not know, so far, how to characterise the continuous maps $\beta(\Mod_{n})\to \V(X)\times X$ arising in this manner, which would establish the \emph{completeness} of the construction. This is a notable obstacle to a full duality theoretic understanding of step-by-step quantification in predicate logics. On the other hand, such a completeness result is available for semiring quantifiers in logic on words, and makes use of the richer structure of the spaces of models (in the form of monoid actions). See Proposition VI.7 and Theorem VI.8 of \citep{GPR2017}, where this is called a `Reutenauer-type theorem'. A question arises, whose answer would significantly further the use of topological methods in logic: \emph{Is there a Reutenauer-type result for first-order logic over arbitrary structures? }

\vspace{1em}
In this paper we have discussed several examples of topological methods in logic and computer science, highlighting their duality theoretic nature. However, there are topological methods in logic which have been successfully developed and applied, but for which no duality theoretic explanation is available so far. An appealing example is the theory of limits of schema mappings as developed in database theory by Kolaitis and his collaborators \citep{Kolaitis2018}. Understanding these tools and results from a duality theoretic perspective may yield new useful insights and is an exciting venue for future investigations. Another example are 0--1 laws in finite model theory, illustrating the limits of the expressive power of first-order logic over finite structures, see e.g.\ \citep{Fagin1976}. These are only some of the many opportunities for further development of the duality approach, which would contribute to unify the `structure' and `power' strands in theoretical computer science.

\vspace{1em}
One of the main themes of our present contribution has been the analysis of step-by-step constructions in logic, which yield \emph{free} objects on the algebra side and \emph{co-free} objects on the space side. Note that, even though the step-wise process of adding a layer of connectives yields a \emph{monad} in the (co)limit, the one-step functor is typically a \emph{comonad}. For instance, the functor on Boolean algebras which adds one layer of modality $\Diamond$ is a comonad, whose dual is the Vietoris monad on Boolean spaces.

The recent work of Samson Abramsky and his coauthors on comonads for model-theoretic games \citep{Abramsky2017b, AbramskyShah2018} is tightly related to this viewpoint. The connection between the comonadic approach and the duality one remains to be explored, and is an interesting avenue of research. In this direction, one may point out that the Ehrenfeucht-Fra{\"i}ss{\'e} comonad introduced by Abramsky and Shah arises as the density comonad for a certain (contravariant) realization functor from a category of primitive positive sentences into the category of structures. 

\vspace{1em}
Besides the inductive treatment of quantifiers, another important theme of this paper has been the lack of freeness of Lindenbaum-Tarski algebras of first-order theories. Indeed, we pointed out that this is one of the main obstacles to a satisfactory algebraic and duality theoretic approach to predicate logics.

Another place where the lack of freeness plays an important role is quantum information and computation, to which Samson Abramsky has greatly contributed. There, as recently observed by Abramsky, the lack of freeness (of certain Boolean subalgebras of partial Boolean algebras) can be regarded as an obstruction to classicality. In fact, in the presence of freeness, the Kochen-Specker theorem does not apply. See \citep{AB2020}. Interestingly, in this context, this obstruction represents a (quantum) advantage.

\vspace{1em}
We conclude with a question concerning a wider issue, which is instrumental in addressing the divide between structure and power, one of the main focuses of Samson Abramsky's recent research. 
A difference between general model theory and finite model theory which is often emphasised is the fact that the major structure theorems such as compactness, L{\"o}wenheim-Skolem, etc.\ do not carry over to the finite setting. Rossman's Finite Homomorphism Preservation Theorem is a major advance because it provides such a theorem which does persist in the finite setting. 
Another take on this would be to conjecture that \emph{topological variants of all the classical structure theorems} hold in the finite setting. A first result in this direction is Reiterman's theorem for finite algebras, which shows that Birkhoff's variety theorem has a finite variant once we topologize. 
\emph{In weaker logics of resources, as studied for example in finite model theory, is there a topological component missing at the level of the associated Lawvere theories/categorical semantics?}

\bibliographystyle{apalike}

\end{document}